\documentclass[11pt,reqno,a4paper]{amsart}
\pdfoutput=1
\usepackage[foot]{amsaddr}

\usepackage{ifthen}
\newboolean{proceedings_version}
\setboolean{proceedings_version}{false}
\usepackage{amsfonts,amstext,amsmath,amssymb,mathrsfs,stmaryrd,calc,amsthm,mathtools,bold-extra,mathdots}
\usepackage{bm}%
\usepackage{enumerate}
\usepackage{macros}

\usepackage{thmtools,thm-restate}
\newtheorem{theorem}{Theorem}
\newtheorem{lemma}[theorem]{Lemma}

\newtheorem{fact}[theorem]{Fact}

\theoremstyle{definition}
\newtheorem{definition}[theorem]{Definition}
\newtheorem{example}[theorem]{Example}
\theoremstyle{remark}

\newtheorem{remark}[theorem]{Remark}
\usepackage{subfig,tikz}
\usetikzlibrary{positioning}
\usetikzlibrary{arrows}
\usetikzlibrary{automata}
\usetikzlibrary{trees}
\tikzstyle{intermediate}=[draw=gray!90,very thick,fill=gray!20,circle]
\tikzstyle{state}=[intermediate,minimum width=.7cm]
\tikzstyle{every node}=[font=\small]
\tikzstyle{every edge}=[draw,->,>=stealth',shorten >=1pt,semithick]
\tikzstyle{accepting}=[accepting by arrow]
\tikzstyle{initial}=[initial by arrow,initial text=]
\usepackage[numbers]{natbib}
\renewcommand{\cite}{\citep}

\usepackage{url,hyperref}

\providecommand{\urlstyle}[1]{}
\providecommand{\doi}[1]{\href{http://dx.doi.org/#1}{\nolinkurl{doi:#1}}}
\begin{document}
\renewcommand{\sectionautorefname}{Section}
\renewcommand{\subsectionautorefname}{Section}
\renewcommand{\subsubsectionautorefname}[1]{\S}
\title{The Power of Well-Structured Systems}
\thanks{Work supported by the ReacHard project, ANR grant 11-BS02-001-01.}
\author[S.~Schmitz]{Sylvain Schmitz}
\author[Ph.~Schnoebelen]{Philippe Schnoebelen}
\address{LSV, ENS Cachan \& CNRS, Cachan, France}
\email{\{schmitz,phs\}@lsv.ens-cachan.fr}
\begin{abstract}
  Well-structured systems, aka WSTSs, are computational models where
  the set of possible configurations is equipped with a
  well-quasi-ordering which is compatible with the transition relation
  between configurations. This structure supports generic decidability
  results that are important in verification and several other fields.

This paper recalls the basic theory underlying well-structured systems
and shows how two classic decision algorithms can be formulated as an
exhaustive search for some ``bad'' sequences. This lets us describe
new powerful techniques for the complexity analysis of WSTS
algorithms. Recently, these techniques have been successful in
precisely characterising the power, in a complexity-theoretical sense,
of several important WSTS models like unreliable channel systems,
monotonic counter machines, or networks of timed systems.
\end{abstract}

\maketitle

\section*{Introduction}

\emph{Well-Structured (Transition) Systems}, aka WSTS, are a family of  computational models where 
the usually infinite set of states is equipped with a well-quasi-ordering that
is ``compatible'' with the computation steps. The existence of this
well-quasi-ordering allows for the decidability of some important
  behavioural properties like Termination or Coverability. 

Historically, the idea can be traced back to~\citet{finkel87c} who
gave a first definition for WSTS abstracting from Petri nets and fifo
nets, and who showed the decidability of Termination and Finiteness
(aka Boundedness). Then \citet{finkel94} applied the WSTS idea to
Termination of lossy channel systems, while \citet{abdulla96b} introduced
the backward-chaining algorithm for Coverability. One will find a good
survey of these early results, and a score of WSTS examples,
 in~\cite{abdulla2000c,finkel98b,abdulla2010,BS-fmsd2012}.

The basic theory saw several important developments in
recent years, like
the study of comparative expressiveness for WSTS~\cite{abdulla2011},
or the completion technique for forward-chaining in
WSTS~\cite{FWD-WSTS-J2}. Simultaneously, many new WSTS models have
been introduced
(in distributed computing, software
verification, or other fields), using well-quasi-orderings based on trees,
sequences of vectors, or graphs (see references in~\cite{SS-esslli2012}),
rather than the more traditional vectors of natural numbers or words
with the subword relation.

Another recent development is the complexity analysis of WSTS models
and algorithms.
New techniques, borrowing
notions from proof theory and ordinal analysis, can now precisely
characterise the complexity of some of the most widely used
WSTS~\cite{lcs,phs-mfcs2010,HSS-lics2012}.  The difficulty here is that
one needs complexity functions, complexity classes, and hard problems,
with which computer scientists are not familiar.

The aim of this paper is to provide a gentle introduction to the
main ideas behind the complexity analysis of WSTS algorithms. 
These
complexity questions are gaining added relevance: more and
more recent papers rely on reductions to (or from) a known WSTS problem to show the
decidability (or the hardness) of problems in unrelated fields, from
modal and temporal logic~\cite{kurucz06,mtl} to XPath-like queries~\cite{jurdzinski2011,barcelo13}.

\subsubsection*{Outline of the paper.}
Section~\ref{sec-wsts} recalls the definition of WSTS, \autoref{sec-broadcast} illustrates it
with a simple example, while \autoref{sec-verif} presents  the two main
verification algorithms for WSTS.  Section~\ref{sec-upb} bounds the running
time of these algorithms by studying the length of bad sequences using
fast-growing functions. Section~\ref{sec-low} explains how lower bounds
matching these enormous upper bounds have been established in a few
recent works, including the $\FC{\ezero}$-completeness result in this
volume~\cite{HaaseSS13}.

\section{What are WSTS?}
\label{sec-wsts}
A simple, informal way to define WSTS is to say that they are
transition systems \emph{whose behaviour is monotonic w.r.t.\  a
  well-ordering}. Here, monotonicity of behaviour means that the
states of the transition system are ordered in a way such that larger
states have more available steps than smaller states. Requiring that
the ordering of states is a well-ordering (more generally, a
well-quasi-ordering) ensures that monotonicity translates into
decidability for some behavioural properties like Termination or
Coverability.

Let us start with monotonicity. In its simplest form, a
\emph{transition system} (a TS) is a structure $\Scal=(S,\to)$ where $S$ is the set
of \emph{states} (typical elements $s_1,s_2,\ldots$) and
${\to}\subseteq S\times S$ is the \emph{transition relation}. As
usual, we write ``$s_1\to{}s_2$'' rather than ``$(s_1,s_2)\in{\to}$''
to denote steps.  A TS is \emph{ordered} when it is
further equipped with a quasi-ordering of its states, i.e., a
reflexive and transitive relation ${\leq}\subseteq S\times S$.
\begin{definition}[Monotonicity]
\label{def-monotonicity}
An ordered transition system $\Scal=(S,{\to},{\leq})$
is \emph{monotonic}
$\equivdef$ for all $s_1,s_2,t_1\in S$
\[
\bigl(s_1\to s_2 \text{ and } s_1\leq t_1\bigr)
\text{ implies }
\exists t_2\in S: \bigl(t_1\to t_2 \text{ and } s_2\leq t_2\bigr)
\:.
\]
\end{definition}
This property is also called ``\emph{compatibility}'' (of the ordering
with the transitions)~\cite{finkel98b}. Formally, it just means that
$\leq$ is a \emph{simulation} relation for $\Scal$, in precisely the
classical sense of \citet{milner90b}. The point of
Def.~\ref{def-monotonicity} is to ensure that a ``larger state'' can
do ``more'' than a smaller state. For example, it entails the
following Fact that plays a crucial role in \autoref{sec-verif}.

Given two finite runs $\sss=(s_0\to s_1\to \cdots \to s_n)$ and $\ttt=(t_0\to
 t_1\to \cdots \to t_m)$, we say that $\sss$ simulates
 $\ttt$ \emph{from below} (and $\ttt$ simulates $\sss$ \emph{from
 above}) if $n=m$ and $s_i\leq t_i$ for all $i=0,\ldots,n$.
\begin{fact}
\label{fact-repeat}
Any run $\sss=(s_0\to s_1\to \cdots\to s_n)$ in a WSTS $\Scal$ can be
simulated from above, starting from any $t\geq s_0$.
\end{fact}

\begin{remark}
Definition~\ref{def-monotonicity} comes in many variants. For
example, \citet{finkel98b} consider \emph{strict
compatibility} (when $<$, the strict ordering underlying $\leq$, is a
simulation), \emph{transitive compatibility} (when $\leq$ is a weak
simulation), and the definition can further extend to \emph{labelled}
transition systems. These are all inessential variations of the main
idea.
\qed
\end{remark}

Now to the wqo ingredient.
\begin{definition}[Wqo]
\label{def-wqo}
A quasi-order $(S,\leq)$ is \emph{well} (``is a \emph{wqo}'')
if every infinite sequence  $s_0,s_1,s_2,\ldots$ over $S$ contains an
increasing pair $s_i\leq s_j$ for some $i<j$.
Equivalently, $(S,\leq)$ is a wqo if, and only if, every infinite
sequence $s_0,s_1,s_2,\ldots$ over $S$ contains an infinite increasing
subsequence $s_{i_0}\leq s_{i_1}\leq s_{i_2}\leq\cdots$, where
$i_0<i_1<i_2<\cdots$
\end{definition}
We call \emph{good} a sequence that contains an increasing pair,
otherwise it is \emph{bad}. Thus in a wqo all infinite sequences are
good, all bad sequences are finite.
\\

Definition~\ref{def-wqo} offers two equivalent definitions.
Many other characterisations exist~\cite{kruskal72}, and it is an
enlightening exercise to prove their equivalence
\citep[see][Chap.~1]{SS-esslli2012}.
Let us illustrate the usefulness of the alternative definition: for a
dimension $k\in\Nat$, write $\Nat^k$ for the set of $k$-tuples, or
vectors, of natural numbers. For two vectors $\aaa=(a_1,\ldots,a_k)$
and $\bbb=(b_1,\ldots,b_k)$ in $\Nat^k$, we let $\aaa\leq_\times\bbb$
$\equivdef$ $a_1\leq b_1\land \cdots\land a_k\leq b_k$.
\begin{example}[Dickson's Lemma]
\label{lem-dickson}
$(\Nat^k,\leq_\times)$ is a wqo.
\end{example}
\begin{proof}[Proof of Dickson's Lemma]
Consider an infinite sequence $\aaa_1,\aaa_2,\aaa_3,\ldots$ over
$\Nat^k$ and write $\aaa_i=(a_{i,1},\ldots,a_{i,k})$. One can extract
an infinite subsequence $\aaa_{i_1},\aaa_{i_2},\aaa_{i_3},\ldots$ that
is increasing over the first components, i.e., with $a_{i_1,1}\leq
a_{i_2,1}\leq a_{i_3,1}\leq \cdots$, since $(\Nat,\leq)$ is a wqo
(easy to prove, here the first definition suffices). From this infinite
subsequence, one can further extract an infinite subsequence that is
also increasing on the second components (again, using that $\Nat$ is
wqo). After $k$ extractions, one has an infinite subsequence that is
increasing on all components, i.e., that is increasing for
$\leq_\times$ as required.
\end{proof}

We can now give the central definition of this paper:
\begin{definition}[WSTS]
\label{def-wsts}
An ordered transition system $\Scal=(S,\to,\leq)$ is a WSTS
$\equivdef$ $\Scal$ is monotonic and $(S,\leq)$ is a wqo.
\end{definition}

\section{A Running Example: Broadcast Protocols}
\label{sec-broadcast}

As a concrete illustration of the principles behind WSTS, let us
consider distributed systems known as \emph{broadcast
  protocols}~\citep{emerson98,esparza99}.  Such systems gather an
unbounded number of identical finite-state processes running
concurrently, able to spawn new processes, and communicating either
via \emph{rendez-vous}---where two processes exchange a message---or
via \emph{broadcast}---where one process sends the same message to
every other process.  While this may seem at first sight rather
restricted for modelling distributed algorithms, broadcast protocols
have been employed for instance to verify the correction of cache
coherence protocols without fixing a number of participating
processes.\footnote{See also the up-to-date survey of parametrised
  verification problems in \citep{esparza13}.}

Formally, a broadcast protocol is defined as a triple $\Btt=(Q,M,R)$
where $Q$ is a finite set of \emph{locations}, $M$ a finite set
of \emph{messages}, and $R$ is a set of \emph{rules}, that is, tuples
$(q,\mathit{op},q')$ in $Q\times\mathit{Op}\times Q$, each describing an
operation $\mathit{op}$ available in the location $q$ and leading to a
new location $q'$, where $\mathit{op}$ can be a sending (denoted $m{!}$) or a
receiving ($m{?}$) operation of a rendez-vous message $m$ from $M$, or
a sending ($m{!!}$) or receiving ($m{?\!?}$) operation of a broadcast
message $m$ from $M$, or a spawning ($\text{sp}(p)$) of a new
process that will start executing from location $p$.  As usual, we
write $q\xrightarrow{\mathit{op}}_{\Btt}q'$ if $(q,\mathit{op},q')$ is
in $R$.

Figure~\ref{fig-broadcast} displays a toy example where
$Q=\{r,c,a,q,\bot\}$ and $M=\{d,m\}$: processes in location $c$ can
spawn new ``active'' processes in location $a$, while also moving to
location $a$ (a rule depicted as a
double arrow in \autoref{fig-broadcast}).  These active processes are
flushed upon receiving a broadcast of either $m$ (emitted by a process
in location $q$) or $d$ (emitted by a process in location $r$);
location $\bot$ is a sink location modelling process destruction.

\begin{figure}[tbp]
  \centering
  \begin{tikzpicture}[on grid,auto,node distance=2cm]
    \node[state](r){$r$};
    \node[state,right=of r](c){$c$};
    \node[state,below=of c](a){$a$};
    \node[state,right=of c](q){$q$};
    \node[state,right=of q](p){$\bot$};
    \node[below right=1cm and .2cm of c](s){};
    \path
    (r) edge node{$d{!!}$} (c)
    (c) edge[-,bend left=5] (s.north)
    (s.north) edge[bend left=15] (a)
    (s.north) edge[bend right=15] (a)
    (a) edge[bend left=20] node{$m{?\!?}$} (c)
    (c) edge node{$d{?\!?}$} (q)
    (q) edge node{$m{!!}$} (p);
  \end{tikzpicture}
  \caption{A broadcast protocol.\label{fig-broadcast}}
\end{figure}
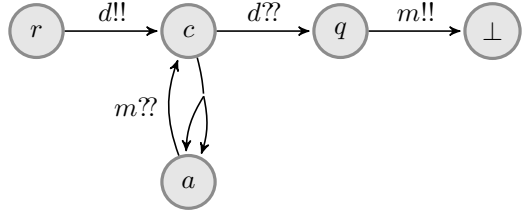

The operational semantics of a broadcast protocol is
expressed as a transition system $\Scal_\Btt=(S,\to)$, where states,
here called configurations, are (finite) multisets of locations in
$Q$, hence $S=\Nat^Q$.  Informally, the intended semantics for a
configuration $s$ in $\Nat^Q$ is to record for each location $q$ in
$Q$ the number of processes $s(q)$ currently in this location.  We use
 a ``sets with duplicates'' notation, like $s=\mset{q_1,\dots,q_n}$ where
some $q_i$'s might be identical, and feel free to write, e.g.,
$\mset{q^3,q'{}^4}$ instead of $\mset{q,q,q,q',q',q',q'}$.
A natural ordering for $\Nat^Q$ is the \emph{inclusion} ordering defined by
$s\subseteq s'\equivdef \forall q\in Q,\,s(q)\leq s'(q)$.  For
instance, $\mset{q,q,q'}\subseteq\mset{q,q,q,q'}$ but
$\mset{q,q',q'}\not\subseteq\mset{q,q,q'}$ if $q\neq q'$.
We further write $s=s_1 + s_2$ for the union (not the lub) of two multisets,
in which case $s-s_1$ denotes $s_2$.

It remains to define how the operations of $\Btt$ update such
a configuration through transitions $s\to s'$ of $\?S_\Btt$:
\begin{description}
\item[rendez-vous step] if
$q_1\xrightarrow{m{!}}_{\Btt}q'_1$ and
$q_2\xrightarrow{m{?}}_{\Btt}q'_2$ for some $m\in M$, then
$s+\mset{q_1,q_2}\to s+\mset{q'_1,q'_2}$  for all $s$ in $\+N^Q$,

\item[spawn step] if $q\xrightarrow{\text{sp}(p)}_{\Btt}q'$, then
$s+\mset{q}\to s+\mset{q',p}$ for
all $s$ in $\+N^Q$,

\item[broadcast step] if
$q_0\xrightarrow{m{!!}}_{\Btt}q'_0$ and
$q_i\xrightarrow{m{?\!?}}_{\Btt}q'_i$ for all $1\leq i\leq k$ (and some
$m$), then
$s+\mset{q_0,q_1,\dots,q_k}\to s+\mset{q'_0,q'_1,\dots,q'_k}$
for all $s$ in $\+N^Q$ that do not contain a potential receiver for the
broadcast, i.e., such that
$s(q)=0$ for all rules of the form  $q\xrightarrow{m{?\!?}}_{\Btt}q'$.

\end{description}

With the protocol of \autoref{fig-broadcast}, the following steps are
possible (with the spawned location or exchanged messages indicated on
the arrows):
\begin{align*}
  \mset{c^2,q,r}&\xrightarrow{a}\mset{a^2,c,q,r}\xrightarrow{a}\mset{a^4,q,r}\xrightarrow{m}\mset{c^4,r,\bot}\xrightarrow{d}\mset{c,q^4,\bot}.
\end{align*}

We have just associated an ordered transition system
$\Scal_\Btt=(\Nat^Q,\to,\subseteq)$ with every broadcast protocol
$\Btt$ and are now ready to prove the following fact.
\begin{fact}
Broadcast protocols are WSTS.
\end{fact}
\begin{proof}
First, $(\Nat^Q,\subseteq)$ is a wqo: since $Q$ is
finite, this is just another instance of Dickson's Lemma.

There remains to check that $\Scal_\Btt$ is monotonic. Formally, this
is done by considering an arbitrary step $s_1\to s_2$ (there are three
cases) and an arbitrary pair $s_1\subseteq t_1$. It is enough to
assume that $t_1=s_1+\mset{q}$, i.e., $t_1$ is just one location
bigger that $s_1$, and to rely on transitivity. If $s_1\to s_2$ is a
rendez-vous step with $s_2=s_1-\mset{q_1,q_2}+\mset{q'_1,q'_2}$, then
$t_1=s_1+\mset{q}$ also has a rendez-vous step $t_1\to
t_2=t_1-\mset{q_1,q_2}+\mset{q'_1,q'_2}$ and one sees that
$s_2\subseteq t_2$ as required. If now if $s_1\to s_2$ is a spawn
step, a similar reasoning proves that $s_1+\mset{q}\to s_2+\mset{q}$.
Finally, when $s_1=s+\mset{q_1,\ldots}\to s_2=s+\mset{q'_1,\ldots}$ is
a broadcast step, one proves that $s_1+\mset{q}\to s_2+\mset{q'}$ when
there is a rule $q\xrightarrow{m{?\!?}}_{\Btt}q'$, or when $q$ is not
a potential receiver and $q'=q$.
\end{proof}
  One can show that the protocol depicted in \autoref{fig-broadcast}
  always terminates, this starting from any initial configuration.
 Indeed, consider any sequence of steps $s_0\to s_1\to
  \cdots\to s_i\to\cdots$, write each configuration under the form
  $s_i=\mset{a^{n_{a,i}},c^{n_{c,i}},q^{n_{q,i}},r^{n_{r,i}},\bot^{n_{\bot,i}}}$, and
   compare any two $s_i$ and $s_j$ with $i<j$:
  \begin{itemize}
  \item either only spawn steps occur along the segment $s_i\to s_{i+1}\to \cdots\to
  s_j$,
  thus $n_{c,j}<n_{c,i}$,
  \item or at least one $m$ has been broadcast but no $d$ has been
  broadcast, thus $n_{q,j}<n_{q,i}$,
  \item or at least one $d$ has been broadcast, and then
  $n_{r,j}<n_{r,i}$.
  \end{itemize}
  Thus in all cases, $s_i\not\subseteq s_j$, i.e.\ the sequence 
$s_0,s_1,\ldots$ is bad.  Since $(\+N^Q,\subseteq)$ is a
  wqo, there is no infinite bad sequence, hence no infinite run.

\begin{remark}
  The above property is better known as \emph{structural termination}, or
  also \emph{uniform termination}, i.e., ``termination from any starting
  configuration,'' and is sometimes confusingly called just ``termination''
  in the context of broadcast protocols. Structural termination is
  undecidable for broadcast protocols and related
  models~\citep{esparza99,mayr03,phs-rp10}. By contrast, the usual meaning
  of \emph{termination} is, for a TS $\Scal$ and a given starting state
  $s_\init\in S$, whether all runs starting from $s_\init$ are finite. For
  broadcast protocols, and WSTS in general, termination will be shown
  decidable in \autoref{ssec-termination}.\hfill\qed
\end{remark}

\section{Verification of WSTS}
\label{sec-verif}
In this section we present the two main generic decision algorithms
for WSTS. We strive for a presentation that abstracts away from
implementation details, and that can directly be linked to the
complexity analysis we describe in the following sections. The general
idea is that these algorithms can be seen as an exhaustive search for
some kind of bad sequences.

\subsection{Termination}
\label{ssec-termination}

There is a generic algorithm deciding Termination on WSTS. The
algorithm has been adapted and extended to show decidability of
Inevitability (of which Termination is a special case), Finiteness
(aka Boundedness), or Regular Simulation,
see~\cite{abdulla2000c,finkel98b}.

\begin{lemma}[Finite Witnesses for Infinite Runs]
\label{lem-fin-wit-termin}
A WSTS $\Scal$ has an infinite run from $s_\init$ if, and only if, it has
a finite run from $s_\init$ that is a good sequence.
\end{lemma}
\begin{proof}
Obviously, any infinite run $s_\init=s_0\to s_1\to s_2\to \cdots$ is a good
sequence by property of $\leq$ being a wqo (see Def.~\ref{def-wqo}).
Once a pair $s_i\leq s_j$ is identified, the finite prefix that stops
at $s_j$ is both a finite run and a good sequence.

Reciprocally, given a finite run $s_0\to s_1 \to \cdots \to s_i \to
\cdots \to s_j$ with $i<j$ and $s_i\leq s_j$, Fact~\ref{fact-repeat}
entails the existence of a run $s_j\to s_{j+1}\to \cdots \to s_{2j-i}$
that simulates $s_i\to \cdots \to s_j$ from above. Hence the finite
run can be extended to some $s_0\to\cdots\to s_i\to\cdots\to s_j\to
\cdots\to s_{2j-i}$ with $s_i\leq s_{2j-i}$. Repeating this extending
process ad infinitum, one obtains an infinite run.
\end{proof}

Very little is needed to turn Lemma~\ref{lem-fin-wit-termin} into a
decidability proof for Termination. We shall make some
minimal \emph{effectiveness assumptions}: (EA1) the set of states $S$ is
recursive; (EA2) the function $s\mapsto \Post(s)$, that associates with
any state its image by the relation $\to$, is
computable (and image-finite, aka finitely branching); and (EA3) the
wqo $\leq$ is decidable. We say that $\Scal$ is an \emph{effective}
WSTS when all three assumptions are fulfilled. Note that (EA1--2) hold
of most computational models, starting with Turing machines and
broadcast protocols, but we have to spell out these assumptions at some point since
Def.~\ref{def-wsts} is abstract and does not provide any algorithmic
foothold. %

We can now prove the decidability of Termination for effective WSTS.
Assume $\Scal=(S,\to,\leq)$ is effective. We are given some starting
state $s_\init\in S$. The existence of an infinite run is
semi-decidable since infinite runs admit finite witnesses by
Lemma~\ref{lem-fin-wit-termin}. (Note that we rely on all three
effectiveness assumptions to guess a finite sequence
$(s_\init=)s_0,\ldots,s_i,\ldots,s_j$ of states, check that it is
indeed a run of $\Scal$, and that it is indeed a good sequence.)
Conversely, if all runs from $s_\init$ are finite, then there are only
finitely many of them (by K\H{o}nig's Lemma, since $\Scal$ is finitely
branching), and it is possible to enumerate all these runs by
exhaustive simulation, thanks to (EA2). Thus Termination is
semi-decidable as well. Finally, since the Termination problem and its
complement are both semi-decidable, they are decidable.

\subsection{Coverability}
\label{ssec-covering}

After Termination, we turn to the decidability of Coverability, a
slightly more involved result that is also more useful for practical
purposes: the decidability of Coverability opens the way to the
verification of safety properties and many other properties defined by
fixpoints, see~\cite{BS-fmsd2012}.

Recall that, for an ordered TS $\Scal$, Coverability is the question, 
given a starting state $s_\init\in S$ and a target state
$t\in S$, whether there is a run from $s_\init$ that eventually covers
$t$, i.e., whether there is some $s$ reachable from $s_\init$ with
$s\geq t$.  We call any finite run $s_0\to s_1\to \cdots \to s_n$
s.t.\ $s_n\geq t$ a \emph{covering run (for $t$)}.

Rather than using covering runs to witness Coverability, we shall use
``pseudoruns''. Formally, a \emph{pseudorun} is a sequence
$s_0,\ldots,s_n$ such that, for all $i=1,\ldots,n$, $s_{i-1}$ can
cover $s_i$ \emph{in one step}, i.e., $s_{i-1}\to t_i$ for some
$t_i\geq s_i$. In particular, any run is also a pseudorun. And the
existence of a pseudorun $s_0,\ldots,s_n$ with $s_n\geq t$ witnesses
the existence of covering runs from any $s_\init\geq s_0$
(proof: by repeated use of Fact~\ref{fact-repeat}).

A pseudorun $s_0,\ldots,s_n$ is \emph{minimal} if, for all
$i=1,\ldots,n$, $s_{i-1}$ is minimal among all the states from where
$s_i$ can be covered in one step (we say that $s_{i-1}$ is a \emph{minimal
pseudopredecessor} of $s_i$).
\begin{lemma}[Minimal Witnesses for Coverability]
\label{lem-cover}
If $\Scal$ has a covering run from $s_\init$, it has in particular a
minimal pseudorun $s_0,s_1,\ldots,s_n$ with $s_0\leq s_\init$,
$s_n=t$, and such that the \emph{reverse sequence}
$s_n,s_{n-1},\ldots,s_0$ is bad (we say that $s_0,\ldots,s_n$ is ``revbad'').
\end{lemma}
\begin{proof}
Assume that $s_\init=s_0\to s_1\to \cdots\to s_n$ is a covering run.
Replacing $s_n$ by $t$ gives a pseudorun ending in $t$. We now show
that if the pseudorun is not minimal or not revbad,
then there is a ``smaller'' pseudorun.

First, assume that $s_n,s_{n-1},\ldots,s_0$ is not bad. Then $s_i\geq
s_j$ for some $0\leq i<j\leq n$ and
$s_0,s_1,\ldots,s_{i-1},s_j,s_{j+1},\ldots,s_n$ is again a pseudorun,
shorter in length (note that $s_n=t$ is unchanged, while $s_0$ may
have been replaced by a smaller $s_j$ in the case where $i=0$). If now
$s_0,s_1,\ldots,s_n$ is not minimal, i.e., if some $s_{i-1}$ is not a
minimal pseudopredecessor of $s_i$, we may replace $s_{i-1}$ by some
other pseudopredecessor $s'_{i-1}\leq s_{i-1}$ that is minimal (since
$(S,\leq)$ is wqo, hence well-founded, its non-empty subsets do have
minimal elements) and $s_0,\ldots,s_{i-2},s'_{i-1},s_i,\ldots,s_n$ is
again a pseudorun (where $s_n=t$ as before and where $s_0$ may have
been replaced by a smaller $s'_0$). Repeating such shortening and
lowering replacements as long as possible is bound to terminate (after
at most polynomially many replacements). The pseudorun we end up with
is minimal, revbad, has $s_0\leq s_\init$ and $s_n=t$ as claimed.
\end{proof}
Turning Lemma~\ref{lem-cover} into a decidability proof is similar to
what we did for Termination. This time we make the following
effectiveness assumptions: (EA1) and (EA3) as above, with (EA2') the
assumption that the function $\MinPPre$---that associates with any
state its finite set of minimal pseudopredecessors---is
computable.\footnote{Recall that any subset of a wqo has only finitely
  many minimal elements up to the equivalence given by $s\equiv
  s'\equivdef s\leq s'\leq s$.} The set of all minimal revbad
pseudoruns ending in $t$ is finite (K\H{o}nig's Lemma again: finite
branching of the tree is ensured by minimality of the pseudoruns,
while finite length of the branches is ensured by the restriction to
revbad pseudoruns). This set of pseudoruns can be built effectively,
starting from $t$ and applying $\MinPPre$ repeatedly, but it is enough
to collect the states that occur along them, using a standard
backward-chaining scheme. We write $\MinPPre^\ast(t)$ to denote the
set of all these states: once they have been computed, it only remains
to be checked whether $s_\init$ is larger than one of them, using
(EA3) once more.

We conclude by observing that assumption (EA2'), though less
natural-looking than (EA2), is satisfied in most computational models.
As Ex.~\ref{ex-control-suc} shows for the case of broadcast protocols,
computing $\MinPPre(s)$ is often a simple case of finding the minimal
solutions to a simple inverse problem on rewrite rules.

\subsection{What is the Complexity of WSTS Verification?}

\label{ssec-ex-broadcast}
  One aspect of the algorithms given in this section we have swept
  under the rug is how expensive they can be.  This will be the topic
  of the next two sections, but as an appetiser, let us consider how
  long the Termination algorithm can run on the broadcast protocol
  of \autoref{fig-broadcast}.

  Let us ignore the number of processes in the sink location $\bot$.
  The protocol from \autoref{fig-broadcast} allows the following steps
  when spawn steps are performed as long as possible before
  broadcasting $m$:
  \begin{align*}
    \mset{c^n,q}&\xrightarrow{a^n}\mset{a^{2n},q}\xrightarrow{m}\mset{c^{2n}}\;.
    \intertext{Such a greedy sequence of message thus doubles the
      number of processes in $c$ and removes one single process from
      $q$.  Iterating such sequences as long as some process is in $q$
      before broadcasting $d$ then leads to:}
    \mset{c^{2^0},q^n,r}&\xrightarrow{a^{2^0}m}\mset{c^{2^1},q^{n-1},r}\xrightarrow{a^{2^1}m}\mset{c^{2^2},q^{n-2},r}\\&\cdots\to\mset{c^{2^{n-1}},q,r}\xrightarrow{a^{2^{n-1}}m}\mset{c^{2^n},r}\xrightarrow{d}\mset{c^{2^0},q^{2^n}}\;.
    \intertext{Such iterations thus implement an exponentiation
      of the number of processes in $q$ in exchange for decrementing
      the number of processes in $r$ by one.  Repeating this kind of
      sequences therefore allows:}
    \mset{c,q,r^n}&\to^\ast\mset{c,q^{\text{tower}(n)}}\:,
  \end{align*}
  where $\text{tower}(0)\eqdef 1$ and $\text{tower}(n+1)\eqdef
  2^{\text{tower}(n)}$: although it always terminates, the broadcast
  protocol of \autoref{fig-broadcast} can exhibit sequences of steps
  of non-elementary length.  This also entails a non-elementary lower
  bound on the Termination algorithm when run on this protocol: since
  the system terminates, all the runs need to be checked, including
  this particular non-elementary one.
  \qed

\section{Upper Bounds on Complexity}
\label{sec-upb}
Although the theory of well-structured  systems provides 
generic algorithms for numerous verification problems, it might
seem rather unclear, what the computational cost of running these
algorithms could be---though we know their complexity can be
considerable (recall \autoref{ssec-ex-broadcast}).  Inspecting the
termination arguments in \autoref{sec-verif}, we see that the critical
point is the finiteness of bad sequences.  Unfortunately, the wqo
definition does not mention anything about the \emph{length} of such
sequences, but merely asserts that they are finite.

It turns out that very broadly applicable hypotheses suffice in order
to define a maximal length for bad sequences (\autoref{sub-control}),
which then gives rise to so-called \emph{length function theorems}
bounding such lengths using ordinal-indexed functions
(\autoref{sub-lft}).  These upper bounds allow for a classification of
the power of many  WSTS models in complexity-theoretic terms
(\autoref{sub-fgcc}), and also lead to simplified WSTS algorithms
that take advantage of the existence of computable upper bounds on the
length of bad sequences (\autoref{sub-comb-algo}).

\subsection{Controlled Sequences}\label{sub-control}

\subsubsection{The Length of Bad Sequences.}
If we look at a very simple quasi-order, namely $(Q,{=})$ with a
finite support $Q$ and equality as ordering---which is a wqo by the
pigeonhole principle---, we can only exhibit bad sequences with length
up to $\#_{\!Q}$, the cardinality of $Q$.  But things start going awry as soon
as we consider infinite wqos; for instance
\begin{equation}\label{eq-bad-nat}
\tag{S1}
  n, n-1, n-2, \dots, 0
\end{equation}
is a bad sequence over $(\+N,{\leq})$ for every $n$ in $\+N$, i.e.\
the length of a bad sequence over $(\+N,{\leq})$ can be arbitrary.
Even if we restrict ourselves to bad sequences where the first element
is not too large, we can still build arbitrarily long sequences: for
instance, over $(\+N^Q,\subseteq)$ with $Q=\{p,q\}$,
\begin{equation}\label{eq-bad-Nat2}
\tag{S2}
  \mset{p}, \mset{q^{n}},\mset{q^{n-1}},\dots,\mset{q},\emptyset
\end{equation}
is a bad sequence of length $n+2$.

\subsubsection{Controlling Sequences.}
Here is however a glimpse of hope: the sequence \eqref{eq-bad-Nat2}
cannot be the run of a broadcast protocol, due to the
sudden ``jump'' from $\mset{p}$ to an arbitrarily large configuration
$\mset{q^{n}}$.  More generally, the key insight is that, in an
algorithm that relies on a wqo for termination, successive states
cannot jump to arbitrarily large sizes, because these states are
constructed algorithmically: we call such sequences \emph{controlled}.

Let $(A,\leq_A)$ be a wqo.  Formally, we posit a \emph{norm}
$|.|_A{:}\,A\to\+N$ on the elements of our wqo, which we require to
be \emph{proper}, in that only finitely many elements have norm $n$:
put differently, $A_{\leq n}\eqdef\{x\in A\mid |x|_A\leq n\}$ must be
finite for every $n$.  For instance, $|s|_{\+N^Q}\eqdef\max_{q\in
Q}s(q)$ is a proper norm for $\+N^Q$.

Given an increasing \emph{control function} $g{:}\,\+N\to\+N$, we say
that a sequence $x_0,x_1,\dots$ over $A$ is \emph{$(g,n_0)$-controlled} if the
norm of $x_i$ is no larger that the $i$th iterate of $g$ applied to
$n_0$: $|x_i|_A\leq g^i(n_0)$ for all $i$.  Thus $g$ bounds
the growth of the elements in the sequence, and $n_0$ is a bound on the
initial norm $|x_0|_A$.

\begin{example}[Controlled Successors in a Broadcast Protocol]
  \label{ex-control-suc}
  Let us see how these definitions work on broadcast protocols.
  First, on the sequences of successors built in the Termination
  algorithm: If  $s\to s'$ is a rendez-vous
  step, then $|s'|_{\+N^Q}\leq 2+|s|_{\+N^Q}$, corresponding to the case where
  the two processes involved in the rendez-vous move to the same
  location.  If $s\to s'$ is a broadcast step, then $|s'|_{\+N^Q}\leq
  \#_{\!Q}\cdot |s|_{\+N^Q}$, corresponding to the case where all the processes in
  $s$ (of which there are at most $\#_{\!Q}\cdot |s|_{\+N^Q}$) enter the same
  location.  Finally, spawn steps only incur $|s'|_{\+N^Q}\leq
  |s|_{\+N^Q}+2$.  Thus $g(n)\eqdef\#_{\!Q}\cdot n$ defines a
  control function for any run in a broadcast
  protocol with $\#_{\!Q}\geq 2$ locations, provided the initial norm
  $n_0$ is chosen large enough.\qed
\end{example}

\begin{example}[Controlled Minimal Pseudopredecessors in a Broadcast
  \label{ex-control-pred}
  Protocol] Now for the minimal pseudopredecessors built in the course of
  the Coverability algorithm: %
  Assume $|t|_{\+N^Q}\leq n$ and $s\to s'\geq t$ is
  a step from some minimal $s$:
  \begin{description}
  \item[rendez-vous step]  If $s'=(s-\mset{q_1,q_2}+\mset{q'_1,q'_2})$
  in a rendez-vous, then
  \begin{itemize}
  \item either $\mset{q'_1,q'_2}\subseteq t$ and
  thus $s=t-\mset{q'_1,q'_2}+\mset{q_1,q_2}$ and
  $|s|_{\+N^Q}\leq n$,
  \item or $q'_i\not\in t$ for exactly one $i$ among $\{1,2\}$, hence
  $s=t-\mset{q_{1-i}}+\mset{q_1,q_2}$ and $|s|_{\+N^Q}\leq
  n+1$,
  \item or $q'_i\not\in t$ for any $i\in\{1,2\}$ and $|s|_{\+N^Q}\leq
  n+2$ (note however that $s\subseteq t$ in this case and the
  constructed sequence would not be bad).
  \end{itemize}
  \item[broadcast step]Assume
  $s'\!=(s-\mset{q_0,q_1,\dots,q_k}+\mset{q'_0,q'_1,\dots,q'_k})$ with
  $q_0\!\!\xrightarrow{m{!!}}_{\Btt}q'_0\!$ the corresponding broadcast send
  rule.  Because $s$ is minimal, $\mset{q'_1,\dots,q'_k}\!\subseteq t$,
  as otherwise a smaller $s$ could be used.  Hence,
  \begin{itemize}
  \item either $q'_0\in t$, and then $s=t-\mset{q'_0,q'_1,\dots,q'_k}+\mset{q_0,q_1,\dots,q_k}$ and
  $|s|_{\+N^Q}\leq n$,
  \item or $q'_0\not\in t$, and then
  $s=t-\mset{q'_1,\dots,q'_k}+\mset{q_0,q_1,\dots,q_k}$ and
  $|s|_{\+N^Q}\leq n+1$.
  \end{itemize}
  \item[spawn step] similar to a rendez-vous step, $|s|_{\+N^Q}\leq n+1$.
  \end{description}
  Therefore, $g(n)\eqdef n+2$ defines a control function for any
  sequence of minimal pseudopredecessor steps in any broadcast protocol.\qed
\end{example}

\subsubsection{Length Functions.}
The upshot of these definitions is that, unlike in the uncontrolled
case, there is a longest $(g, n_0)$-controlled bad sequence over any
normed wqo $(A,\leq_A)$: indeed, we can organise such sequences in a tree
by sharing common prefixes; this tree has
\begin{itemize}
\item finite branching, bounded by the cardinal of $A_{\leq
    g^i(n_0)}$ for a node at depth $i$, and
\item no infinite branches thanks to the wqo property.
\end{itemize}
By K\H{o}nig's Lemma, this tree of bad sequences is therefore finite,
of some height $L_{g,n_0,A}$ representing the length of the maximal $(g,
n_0)$-controlled bad sequence(s) over $A$. In the following, since we
are mostly interested in this length as a function of the initial norm
$n_0$, we will see this as a \emph{length function} $L_{g,A}(n)$; our
purpose will then be to obtain complexity bounds on $L_{g,A}$
depending on $g$ and $A$.

\subsection{Length Function Theorems}\label{sub-lft}

Now that we are empowered with a suitable definition for the maximal
length $L_{g,A}(n)$ of $(g,n)$-controlled bad sequences over $A$, we
can try our hand at proving \emph{length function theorems}, which
provide constructible functions bounding $L_{g,A}$ for various normed
wqos $(A,\leq)$.  Examples of length function theorems can be found
\begin{itemize}
\item in \citep{mcaloon,clote,FFSS-lics2011,abriola,SS-esslli2012} for
Dickson's Lemma, i.e.\ for $(\+N^Q,\subseteq)$ for some finite
$Q$, which is isomorphic to $(\+N^{\#_{\!Q}},\leq_\times)$,%
\item in \citep{abriola} for $(\?P_{\!f}(\+N^d),\preceq)$ the set of finite
subsets of $\+N^d$ with the majoring ordering defined by $X\preceq
Y\equivdef \forall x\in X,\exists y\in Y, x\leq_\times y$,
\item in \citep{weiermann94,cichon98,SS-icalp2011} for Higman's Lemma, i.e.\ for
$(\Sigma^\ast,\leq_\ast)$ the set of finite sequences over a finite
alphabet $\Sigma$ with the subword embedding $\leq_\ast$,
\item in \citep{weiermann94} for Kruskal's Tree Theorem, i.e.\ for $(T,\leq_T)$
the set of finite unranked ordered trees with the homeomorphic
embedding $\leq_T$.
\end{itemize}
These theorems often differ in the hypotheses they put on $g$, the
tightness of the upper bounds they provide, and on the simplicity of
their proofs (otherwise the results of \citet{weiermann94} for
Kruskal's Tree Theorem would include all the others).  We will try to
convey the flavour of the theorems
from \citep{SS-icalp2011,SS-esslli2012} here.

Starting again with the case of a finite wqo $(Q,=)$, and setting
$|x|_Q\eqdef 0$ for all $x$ in $Q$ as the associated norm, we find
immediately that, for all $g$ and $n$,
\begin{align}
  L_{g,Q}(n) &= \#_{\!Q}
\intertext{by the pigeonhole principle.  Another easy example is $(\+N,\leq)$
with norm $|k|_{\+N}\eqdef k$:}\label{eq-badd-nat}
  L_{g,\+N}(n) &= n+1\;,
\end{align}
the bad sequence \eqref{eq-bad-nat} being maximal.

We know however from \autoref{ssec-ex-broadcast} that very long bad
sequences can be constructed, so we should not hope to find such
simple statements for $(\+N^Q,\subseteq)$ and more complex wqos.
The truth is that the ``$\text{tower}$'' function we used in
\autoref{ssec-ex-broadcast} is really benign compared to the kind of
upper bounds provided by length function theorems.  Such functions of
enormous growth have mainly been studied in the context of
subrecursive hierarchies and reverse mathematics (see
e.g.\ \citep[Chap.~4]{sw12} for further reference); let us summarily
present them.

\subsubsection{Ordinal Indexed Functions.}
An idea in order to build functions of type $\+N\to\+N$ with faster
and faster growths is to iterate smaller functions a number of times
that depends on the argument---this is therefore a form of
diagonalisation.  In order to keep track of the diagonalisations, we
can index the constructed functions with ordinals, so that
diagonalisations occur at limit ordinals.

\paragraph{\textit{Ordinal Terms.}}
First recall that ordinals $\alpha$ below $\ezero$ can be denoted as terms
in \emph{Cantor Normal Form}, aka CNF:
\begin{align*}%
  \alpha = \omega^{\beta_1}\cdot c_1+\cdots
  +\omega^{\beta_n}\cdot c_n&\text{ where }\alpha > \beta_1
  >\cdots> \beta_n\text{ and }\omega>c_1,\ldots,c_n>0\;.%
\end{align*}
In this representation, $\alpha=0$ if and only if $n=0$.  An ordinal
with {CNF} of the form $\alpha'+1$ (i.e.\ with $n>0$ and $\beta_n=0$)
is called a \emph{successor} ordinal, and otherwise if $\alpha>0$ it
is called a \emph{limit} ordinal, and can be written as
$\gamma+\omega^\beta$ by setting $\gamma=\omega^{\beta_1}\cdot
c_1+\cdots +\omega^{\beta_n}\cdot(c_n-1)$ and $\beta=\beta_n$.  We
usually write ``$\lambda$'' to denote a limit ordinal.

A \emph{fundamental sequence} for a limit ordinal $\lambda$ is a
sequence $(\lambda(x))_{x<\omega}$ of ordinals with supremum
$\lambda$, with a standard assignment defined inductively by
\begin{align}
  \label{eq-fund-def}
  (\gamma+\omega^{\beta+1})(x)&\eqdef\gamma+\omega^\beta\cdot (x+1)\;,&
  (\gamma+\omega^{\lambda})(x)&\eqdef\gamma+\omega^{\lambda(x)}\;.
\end{align}
This is one particular choice of a fundamental sequence, which
verifies e.g.\ $0<\lambda(x)<\lambda(y)$ for all $x<y$.  For
instance, $\omega(x)=x+1$,
$(\omega^{\omega^4}+\omega^{\omega^3+\omega^2})(x)=\omega^{\omega^4}+\omega^{\omega^3+\omega\cdot
  (x+1)}$.

\paragraph{\textit{Hardy Hierarchy.}} Let $h{:}\,\+N\to\+N$ be an
increasing function.  The \emph{Hardy hierarchy}
$(h^\alpha)_{\alpha<\ezero}$ controlled by $h$ is defined
inductively by
\begin{align}\label{eq-hardy-rel}
  h^0(x)&\eqdef x\;,&
  h^{\alpha+1}(x)&\eqdef h^\alpha\left(h(x)\right),&
  h^\lambda(x)&\eqdef h^{\lambda(x)}(x)\;.
\end{align}
Observe that $h^k$ for some finite $k$ is the $k$th iterate of $h$ (by
using the first two equations solely).  This intuition carries over:
$h^\alpha$ is a transfinite iteration of the function $h$, using
diagonalisation to handle limit ordinals.  For instance, starting with
the successor function $H(x)\eqdef x+1$, we see that a first
diagonalisation yields $H^\omega(x)=H^{x+1}(x)=2x+1$.  The next
diagonalisation occurs at $H^{\omega\cdot
  2}(x)=H^{\omega+x+1}(x)=H^{\omega}(2x+1)=4x+3$.  Fast-forwarding a
bit, we get for instance a function of exponential growth
$H^{\omega^2}(x)=2^{x+1}(x+1)-1$, and later a non elementary function
$H^{\omega^3}$, an ``Ackermannian'' non primitive-recursive function
$H^{\omega^\omega}$, and an ``hyper-Ackermannian'' non
multiply-recursive function $H^{\omega^{\omega^\omega}}$.  Hardy
functions are well-suited for expressing large iterates of a control
function, and therefore for bounding the norms of elements in a
controlled bad sequence.

\paragraph{\textit{Cicho\'n Hierarchy.}}  A variant of the Hardy
functions is the \emph{Cicho\'n hierarchy}
$(h_\alpha)_{\alpha<\ezero}$ controlled by $h$~\citep{cichon98},
defined by
\begin{align}\label{eq-cichon-rel}
  h_0(x)&\eqdef 0\;,&
  h_{\alpha+1}(x)&\eqdef 1+h_\alpha\left(h(x)\right),&
  h_\lambda(x)&\eqdef h_{\lambda(x)}(x)\;.
\end{align}
For instance, $h_d(x)=d$ for all finite $d$, thus $h_\omega(x)=x+1$
regardless of the choice of the function $h$.  One can check that
$H^\alpha(x)=H_\alpha(x)+x$ when employing the successor function $H$;
in general $h^\alpha(x)\geq h_\alpha(x)+x$ since $h$ is assumed to be
increasing. 

This is the hierarchy we are going to use for our statements of length
function theorems: a Hardy function $h^\alpha$ is used to bound the
maximal norm of an element in a bad sequence, and the corresponding
Cicho\'n function $h_\alpha$ bounds the length of the bad sequence
itself, the two functions being related for all $h$, $\alpha$, and $x$
by
\begin{equation}\label{eq-hardy-cichon}
  h^\alpha(x) = h^{h_\alpha(x)}(x)\;.
\end{equation}

\subsubsection{Length Functions for Dickson's Lemma.}
We can now provide an example of a length function theorem for a
non-trivial wqo: Consider $(\+N^Q,\subseteq)$ and some control
function $g$.  Here is one of the parametric bounds proved in
\citep[Chap.~2]{SS-esslli2012}:\footnote{A more general version of
  \autoref{thm-lft} in \citep{SS-icalp2011} provides $h_{o(A)}(n)$
  upper bounds for $(g,n)$-controlled bad sequences over wqos
  $(A,\leq)$ constructed through disjoint unions, Cartesian products,
  and Kleene star operations, where $o(A)$ is the maximal order type
  of $(A,\leq)$, i.e.\ the order type of its maximal linearization,
  and $h$ is a low-degree polynomial in $g$.  This matches
  \autoref{thm-lft} because $o(\+N^Q)=\omega^{\#_{\!Q}}$.}%
\begin{theorem}[Parametric Bounds for Dickson's Lemma]\label{thm-lft}
  If $x_0,\ldots,x_L$ is a $(g,n)$-controlled bad sequence over
  $(\+N^Q,\subseteq)$ for some finite set $Q$, then $L\leq L_{g,\+N^Q}(n)\leq
  h_{\omega^{\#_{\!Q}}}(n)$ for the function $h(x)\eqdef\#_{\!Q}\cdot g(x)$.
\end{theorem}
\noindent
By Equation \eqref{eq-hardy-cichon}, we also deduce that the norm of
the elements $x_i$ in this bad sequence cannot be larger than
$h^{\omega^{\#_{\!Q}}}(n)$.

A key property of such bounds expressed with Cicho\'n and Hardy
functions is that they are \emph{constructible} with negligible
computational overhead (just apply their definition on a suitable
encoding of the ordinals), which means that we can employ them in
algorithms (see \autoref{sub-comb-algo} for applications).

Given how enormous the Cicho\'n and Hardy functions can grow, it is
reasonable at this point to ask how tight the bounds provided by
\autoref{thm-lft} really are.  At least in the case $\#_{\!Q}=1$, we see
these bounds match \eqref{eq-badd-nat} since $h_\omega(n)=n+1$.  We will
show in \autoref{sec-low} that similarly enormous complexity lower
bounds can be proven for Termination or Coverability problems on
WSTS, leaving only inessential gaps with the upper bounds like
\autoref{thm-lft}.  In fact, we find such parametric bounds to be overly
precise, because we would like to express simple \emph{complexity}
statements about decision problems.

\subsection{Fast Growing Complexity Classes}\label{sub-fgcc}

As witnessed in \autoref{ssec-ex-broadcast} and the enormous upper
bounds provided by \autoref{thm-lft}, we need to deal with
non-elementary complexities.  The corresponding non-elementary
complexity classes %
are arguably missing from most textbooks and references on complexity.
For instance, the  \emph{Complexity
Zoo}\footnote{\url{https://complexityzoo.uwaterloo.ca}}, an otherwise very richly
populated place, features no
intermediate steps between \textsc{Elementary} and the next class,
namely \textsc{Primitive-Recursive} (aka \pr), and a similar gap
occurs between \pr\ and \textsc{Recursive} (aka R).  If we are to
investigate the complexity of decision problems on WSTSs, much more
fine-grained hierarchies are required.

Here, we employ a hierarchy of ordinal-indexed \emph{fast growing}
complexity classes $(\FC\alpha)_\alpha$ tailored
to \emph{completeness} proofs for non-elementary
problems \citep[see][]{Schmitz13}.  When exploring larger
complexities, the hierarchy includes non primitive-recursive classes,
for which quite a few complete problems have arisen in the recent
years, e.g.\ $\FC\omega$
in~\citep{fct,jancar,urquhart99,phs-mfcs2010,figueira12,bresolin2012,lazic13},
$\FC{\omega^\omega}$ in~\citep{lcs,mtl,ata,tsoreach,pepreg,barcelo13},
$\FC{\omega^{\omega^\omega}}$ in~\citep{HSS-lics2012}, and
$\FC{\ezero}$ in~\citep{HaaseSS13}.

Let us define the classes of reductions and of problems we will consider:
\begin{align}\label{eq-fast-class}
 \FGH{\alpha}&\eqdef\bigcup_{c<\omega}\cc{FDTime}\left(H^{\omega^\alpha\cdot
 c}(n)\right)\,,& \F\alpha&\eqdef\!\!\!\!\!\bigcup_{p\in\bigcup_{\beta<\alpha}\FGH{\beta}}\!\!\!\!\!\cc{DTime}\left(H^{\omega^\alpha}(p(n))\right)\,.
\end{align}
The hierarchy of function classes
$(\FGH{\alpha})_{\alpha\geq 2}$ is the \emph{extended Grzegorczyk
hierarchy}~\citep{lob70}, and provides us with classes of
non-elementary reductions: for instance $\FGH{2}$ is the set of
elementary functions, $\bigcup_{\beta<\omega}\FGH{\beta}$ that of
primitive-recursive functions, and $\bigcup_{\beta<\omega^\omega}\FGH{\beta}$
that of multiply-recursive functions.  The
hierarchy of complexity classes $(\FC\alpha)_{\alpha\geq 3}$ features for
instance %
a class %
$\F\omega$
of Ackermannian problems closed under
    primitive-recursive reductions, and a class %
$\F{\omega^\omega}$
of hyper-Ackermannian problems closed under multiply-recursive
reductions.  Intuitively, $\F\omega$-complete problems are not
primitive-recursive, but only \emph{barely} so, and similarly for the
other levels.

\subsection{Combinatorial Algorithms}\label{sub-comb-algo}

\autoref{thm-lft} together with \autoref{ex-control-suc}
(resp.~\ref{ex-control-pred}) provides complexity upper bounds on the
Termination (resp.\ Coverability) algorithm when applied to broadcast
protocols.  Indeed, a nondeterministic program can guess a witness of
(non-) Termination (resp.\ Coverability), which is of length bounded by
$L_{g,\+N^{Q}}(n)+1$, where $g$ was computed in
\autoref{ex-control-suc} (resp.~\ref{ex-control-pred}) and $n$ is the
size of the initial configuration $s_\init$ (resp.\ target
configuration $t$).  By \autoref{thm-lft} such a witness has length
bounded by $h_{\omega^{\#_{\!Q}}}(n)$ for $h(n)=\#_{\!Q}\cdot g(n)$, and
by \eqref{eq-hardy-cichon}, the norm of the elements along this
sequence is bounded by $h^{\omega^{\#_{\!Q}}}(n)$.  Thus this
non-deterministic program only needs space bounded by
$\#_{\!Q}\cdot \log(h^{\omega^{\#_{\!Q}}}(n))\leq H^{\omega^\omega}(p(n+\#_{\!Q}))$ for some
primitive-recursive function $p$.  Hence:
\begin{fact}
  Termination and Coverability of broadcast protocols are in
  $\FC\omega$.
\end{fact}

Thanks to the upper bounds on the length of bad sequences, the
algorithms sketched above are really \emph{combinatorial} algorithms:
they compute a maximal length for a witness and then
nondeterministically check for its existence.  In the case of
Termination, we can even further simplify the algorithm: if a run
starting from $s_\init$ has length $>L_{g,A}(n)$, this run is
necessarily a good sequence and the WSTS does not terminate.

\section{Lower Bounds on Complexity}
\label{sec-low}
When considering the mind-numbing complexity upper bounds that come
with applications of the Length Function Theorems from
\autoref{sec-upb} to the algorithms of \autoref{sec-verif}, a natural
question that arises is whether this is the complexity one gets when
using the rather simplistic Coverability algorithm from
\autoref{ssec-covering}, or whether it is the intrinsic complexity of
the Coverability \emph{problem} for a given WSTS model.

There is no single answer here: for instance, on Petri nets, a
breadth-first backward search for a coverability witness actually
works in \cc{2ExpTime}~\citep{bozzelli11} thanks to bounds on the size
of minimal witnesses due to \citet{rackoff78}; on the other hand, a
depth-first search for a termination witness can require Ackermannian
time~\citep{cardoza76}, this although both problems are
\cc{ExpSpace}-complete.

Nevertheless, in many cases, the enormous complexity upper bounds
provided by the Length Function Theorems are matched by similar lower
bounds on the complexity of the Coverability and Termination
\emph{problems}, for instance for reset/transfer Petri
nets~\citep[$\FC\omega$-complete, see][]{phs-mfcs2010}, lossy channel
systems~\citep[$\FC{\omega^\omega}$-complete, see][]{lcs}, timed-arc
Petri nets~\citep[$\FC{\omega^{\omega^\omega}}$-complete,
  see][]{HSS-lics2012}, or priority channel systems
\citep[$\FC\ezero$-complete, see][]{HaaseSS13}.

Our goal in this section is to present the common principles
behind these $\FC\alpha$-hardness proofs. We will avoid most of the
technical details, relying rather on simple examples to convey the
main points: the interested readers will find all details in the
references.
\\

Let us consider a WSTS model like broadcast protocols or lossy channel
systems.  Without a rich repertoire of $\FC\alpha$-complete problems,
one proves $\FC\alpha$-hardness by reducing, e.g., from the acceptance
problem for a $H^{\omega^\alpha}$-space bounded Minsky machine $M$. In
order to simulate $M$ in the WSTS model at hand, the essential part is
to design a way to compute $H^{\omega^\alpha}(n_0)$ reliably and store
this number as a WSTS state, where it can be used as a working space
for the simulation of $M$.  In all the cases we know, these
computations cannot be performed directly (indeed, our WSTS are not
Turing-powerful), but $\?S_M$, the constructed WSTS, is able to
\emph{weakly} compute such values. This means that $\?S_M$ may produce
the correct value for $H^{\omega^\alpha}(n_0)$ but also
(nondeterministically) some smaller values.  However, the reduction is
able to include a check that the computation was actually correct,
either at the end of the
simulation~\citep{phs-mfcs2010,lcs,HSS-lics2012,HaaseSS13}---by weakly
computing the inverse of $H^{\omega^\alpha}$ and testing through the
coverability condition whether the final configuration is the one we
started with---, or continuously at every step of the
simulation~\citep{lazic13}.

%
%
%
%

\subsection{Hardy Computations}

As an example, when $\alpha<\omega^k$, one may weakly compute
$H^{\alpha}$ and its inverse using a broadcast protocol with $k+O(1)$
locations.  In order to represent an ordinal $\alpha=\omega^{k-1}\cdot
c_{k-1}+\cdots+\omega^0\cdot c_0$ in {CNF}, one can employ a
configuration $s_\alpha=\{p_0^{c_0},\ldots,p_{k-1}^{c_{k-1}}\}$ in a
broadcast protocol having locations $p_0,\ldots,p_{k-1}$ among others.

There remains other issues with ordinal representations in a WSTS
state (see \autoref{sub-robust}), but let us first turn to the
question of computing some $H^\alpha(n)$. The definition of the Hardy
functions is based on very fine-grained steps, and this usually
simplifies their implementations.  We can
reformulate \eqref{eq-hardy-rel} as a rewrite system over pairs
$(\alpha,x)$ of an ordinal and an argument:
\begin{xalignat}{2}
\label{eq-hardy-rewr}
\tag{\ref{eq-hardy-rel}'}
(\alpha+1,x)&\to (\alpha,x+1)\:,
&
(\lambda,x)&\to (\lambda(x),x)\:.
\end{xalignat}
A sequence $(\alpha_0,x_0)\to (\alpha_1,x_1)\to \cdots\to
(\alpha_\ell,x_\ell)$ of such ``Hardy steps'' implements
\eqref{eq-hardy-rel} and maintains $H^{\alpha_i}(x_i)$ invariant.
It must terminates since $\alpha_0>\alpha_1>\cdots$ is decreasing.
When eventually $\alpha_\ell=0$, the computation is over and the
result is $x_\ell=H^{\alpha_0}(x_0)$.

\begin{example}[Hardy Computations in Broadcast Protocols]
  Implementing \eqref{eq-hardy-rewr} in a broadcast protocol requires
  us to consider two cases.  Recognising whether $\alpha_i$ is a successor
  boils down to checking that $c_0>0$ in the CNF. In that case, and
  assuming the  above representation,
  \eqref{eq-hardy-rewr} is implemented by moving one process from
  location $p_0$ to a location $x$ where the current value of
  $x_i$ is stored. The broadcast protocol will need a rule like
  $p_0\xrightarrow{\text{sp}(x)}\bot$.

  Alternatively, $\alpha_i$ is a limit $\gamma+\omega^b$ when
  $c_0=c_1=\ldots=c_{b-1}=0<c_b$. In that case, \eqref{eq-hardy-rewr}
  is implemented by moving one process out of the $p_b$ location, and
  adding to $p_{b-1}$ as many processes as there are currently in
  $x$. This can be implemented by moving temporarily all processes
  in $x$ to some auxiliary $x_\tmp$ location, then putting them back in $x$ one by one, each
  time spawning a new process in $p_{b-1}$.

The difficulty with these steps (and with recognising that $\alpha_i$
is a limit) is that one needs to test that some locations are empty,
an operation not provided in broadcast protocols (adding emptiness
tests would make broadcast protocols Turing-powerful, and would break
the monotonicity of behaviour).  Instead of being tested, these locations can be
forcefully emptied through broadcast steps. This is where the
computation of $H^{\alpha_i}(x_i)$ may err and end up with a smaller
value, if the locations were not empty.\hfill\qed
\end{example}
We refer to~\cite{phs-mfcs2010} or \citep[Chapter~3]{SS-esslli2012}
for a detailed implementation of this scheme using lossy counters
machines: the encoding therein can easily be reformulated as a
broadcast protocol:

\begin{fact}\label{fc-lowb}
Termination and Coverability of broadcast protocols are
$\FC\omega$-hard.
\end{fact}

\begin{remark}[Parametric Broadcast Protocols]
  The lower bound for Coverability in \autoref{fc-lowb} also holds in
  a variant of broadcast protocols with no spawn operations, which
  considers instead \emph{parametric} initial configurations.  Such
  parametric configurations allow to specify that some locations
  might be initially visited by any finite number of processes.  By
  adding a parametric ``pool'' location, a spawn operation can then
  be seen as a rendez-vous with a process from the pool.

  Termination in parametric broadcast protocols is undecidable,
  since we can reduce from the structural termination problem by
  making every location
  parametric~\citep{esparza99,mayr03,phs-rp10}.\hfill\qed
\end{remark}

\subsection{Robust Encodings}\label{sub-robust}

The above scheme for transforming pairs $(\alpha,x)$ according to
Eq.~\eqref{eq-hardy-rewr} can be used with ordinals higher than
$\omega^k$. Ordinals up to $\omega^{\omega^\omega}$ have been encoded
as configurations of lossy channel systems~\cite{lcs}, of timed-arc
nets (up to $\omega^{\omega^{\omega^\omega}}$,
see~\cite{HSS-lics2012}), and of priority channel systems (up to
$\ezero$, see~\cite{HaaseSS13}). The operations one performs on these
encodings are recognising whether an ordinal is a successor or a
limit, transforming an $\alpha+1$ in $\alpha$, and a $\lambda$ in
$\lambda(x)$. Such operations can be involved, depending on the
encoding and the facilities offered by the WSTS:
see \cite{HSS-lics2012} for an especially involved example.

It is usually not possible to perform Hardy steps exactly in the WSTS
under consideration. Hence one is content with weak implementations
that may err when realising a step $(\alpha_i,x_i)\to
(\alpha_{i+1},x_{i+1})$. One important difficulty arises here: it is
not enough to guarantee that any weak step $(\alpha_i,x_i)\to
(\alpha',x')$ has $\alpha'\leq \alpha_{i+1}$ and $x\leq x_{i+1}$. One
further needs $H^{\alpha'}(x')\leq H^{\alpha_{i+1}}(x_{i+1})$, a
property called ``robustness''. Since Hardy functions are in general
not monotone in the $\alpha$ exponent (see~\cite{SS-esslli2012}),
extra care is needed in order to control what kinds of errors are
acceptable when ending up with $(\alpha',x')$ instead of
$(\alpha_{i+1},x_{i+1})$.  We invite the reader to have a look at the
three above-mentioned papers for examples of how these issues can be
solved in each specific case.

\section{Concluding Remarks}
As the claim \emph{Well-Structured Transition Systems Everywhere!} made
in the title of~\citep{finkel98b} has been further justified by twelve
years of applications of WSTS in various fields, the need to better
understand the computational power of these systems has also risen.
This research program is still very new, but it has already
contributed mathematical tools and methodological guidelines for
\begin{itemize}
\item proving upper bounds, based on \emph{length functions theorems}
  that provide bounds on the length of controlled bad sequences.  We
  illustrated this on two algorithms for Coverability and
  Termination in \autoref{sec-upb}, but the same ideas are readily
  applicable to many algorithms that rely on a wqo for their
  termination---and thus not only in a WSTS context---: one merely has
  to find out how the bad sequences constructed by the algorithm are
  \emph{controlled}.
\item establishing matching lower bounds: here our hope is for the
  problems we have proven hard for some complexity class $\FC\alpha$
  to be reused as convenient ``master'' problems in reductions.  Failing
  that, such lower bound proofs can also rely on a reusable framework
  developed in \autoref{sec-low}: our reductions from Turing or Minsky
  machines with bounded resources construct the machine workspace as
  the result of a Hardy computation, thanks to a suitable {robust}
  encoding of ordinals.
\end{itemize}
There are still many open issues that need to be addressed to advance
this program: to develop length function theorems for more wqos, to
investigate different wqo algorithms (like the computation of
upward-closed sets from oracles for membership and vacuity by
\citet{vjgl}), and to populate the catalogue of master $\FC\alpha$-hard
problems, so that hardness proofs do not have to proceed from first
principles and can instead rely on simpler reductions.

We hope that this paper can be used as an  enticing primer for researchers
who have been using WSTS as a decidability tool only, and are now ready to
use them for more precise complexity analyses.

\subsubsection*{Acknowledgements.}
We thank Christoph Haase and Prateek Karandikar for their helpful comments on an
earlier version of this paper.  The remaining errors are of course
entirely ours.

\bibliographystyle{abbrvnat}
\bibliography{conferences,journalsabbr,power}

\end{document}